\documentclass[11pt,4apaper]{article}
\usepackage[utf8]{inputenc}
\usepackage{verbatim}
\usepackage{amsmath,amssymb,array}
\usepackage{amsthm}
\usepackage{float}
\usepackage[algoruled,linesnumbered,noend]{algorithm2e}

\begin{document}

\newfloat{Algorithm}{h}{loa}
\newtheorem{Thm}{Theorem}[section]
\newtheorem{Lem}[Thm]{Lemma}
\newtheorem{Proposition}[Thm]{Proposition}
\newtheorem{Obs}[Thm]{Observation}
\newtheorem{Corollary}[Thm]{Corollary}
\newtheorem{Cl}[Thm]{Claim}
\theoremstyle{definition}
\newtheorem{Def}{Definition}[section]
\newtheorem{PB}{Problem}
\newtheorem{Example}{Example}[section]
\newcommand{\nota}[1]{{\sffamily\bfseries #1}\marginpar{\framebox{$\mathbf{\Leftarrow}$}}}
\newcommand{\grafo}{\mathcal G}
\newcommand{\greedy}{\textsc   Greedy-CC}
\newcommand{\Algo}{\textsc Greedy-reduced-CC}

\title{A PTAS for the Minimum Consensus Clustering Problem with a Fixed Number
  of Clusters}
\author{Paola Bonizzoni\thanks{Dipartimento di Informatica, Sistemistica e
    Comunicazione, Universit\`a degli Studi di Milano-Bicocca,
    Milano, Italy, \texttt{bonizzoni@disco.unimib.it}}
  \and Gianluca Della Vedova\thanks{Dipartimento di Statistica,
    Universit\`a degli Studi di Milano-Bicocca, Milano, Italy, \texttt{gianluca.dellavedova@unimib.it}}
  \and Riccardo Dondi\thanks{Dipartimento di Scienze dei Linguaggi, della
    Comunicazione e degli Studi Culturali, Universit\`a degli Studi di
    Bergamo, Bergamo, Italy, \texttt{riccardo.dondi@unibg.it}}}


\maketitle
\begin{abstract}
The Consensus Clustering problem has been introduced as an effective
way to analyze the results of different  microarray experiments
\cite{filkov04:_heter_data_integ_consen_clust_formal,DBLP:journals/ijait/FilkovS04}.
The problem  consists of looking for a partition that best
summarizes  a set of input partitions (each corresponding to a different microarray experiment)
under a simple and intuitive cost
function. The problem admits polynomial time algorithms on two input
partitions, but is APX-hard on three input partitions.
We investigate the restriction of Consensus Clustering
when the output partition is required to contain at most $k$ sets, giving a
polynomial time approximation scheme (PTAS) while proving the NP-hardness of  this restriction.
\end{abstract}


\section{Introduction}
Microarray data analysis is
a fundamental task in studying genes. Indeed,
microarray experiments provide measures of gene expression
levels under certain experimental conditions,
showing that groups of genes have a similar behavior under
certain conditions.
However, even slightly different experimental conditions may result in
significantly different expression data.
These gene expression patterns
are useful to understand the relations among genes and
could provide information useful for the construction of genetic
networks.
Nowadays the use of microarrays has become widespread and sufficiently cheap
to justify running a large battery of experiments under similar, albeit not
identical, conditions. The integration of the results is therefore the final
computational step needed to obtain a meaningful interpretation of the data.

In
\cite{filkov04:_heter_data_integ_consen_clust_formal,DBLP:journals/ijait/FilkovS04}
a clustering approach to the integration of different experimental
microarray experimental data was introduced.
In the proposed approach, called \textsc{Consensus Clustering},
the genes are represented by elements of a universe set.
The experimental data under certain experimental condition, are represented
as a partition of the universe set, where a set represents
elements (genes) that have similar expression level in the experiment.
The proposed approach then computes the consensus of the
partitions given by a collection of gene expression data,
since integrating different experimental data is
potentially more informative than the individual experimental data.
More precisely, \textsc{Consensus Clustering} asks for a partition of the universe set
that better summarizes a set of input partitions on the same universe.
The \textsc{Consensus Clustering}  problem
has been studied extensively in the literature
and
its NP-hardness over general instances is well-known
\cite{krivanek86:_hard_probl_hierar_tree_clust,wakabayashi98}.

The minimization version of \textsc{Consensus Clustering}, called
\textsc{Minimum Consensus Clustering},  admits
a $\frac{3}{2}$-approximation algorithm~\cite{DBLP:journals/jacm/AilonCN08} as
well as a
number of heuristics based on cutting-plane
\cite{grotschel89:_cuttin_plane_algor_clust_probl} and simulated
annealing \cite{DBLP:journals/ijait/FilkovS04}.
In the latter paper, it was observed that the problem
is trivially solvable for instances of at most two partitions,
while an open question, as recently recalled \cite{DBLP:journals/jacm/AilonCN08},
is the computational complexity of the problem (for both minimization and maximization
versions) on $k$ input partitions,  for any constant $k > 2$. The question
has been settled in \cite{DBLP:journals/jcss/BonizzoniVDJ08} by showing that
\textsc{Minimum Consensus Clustering} is APX-hard even on instances
with three input partitions, hence making
hopeless the search for a polynomial time algorithm. In this paper we will focus on the restriction of
the problem where the \emph{desired} consensus partition has at most $k$ sets, with
$k$ a constant.

A problem closely related to \textsc{Minimum Consensus Clustering} is
\textsc{Minimum Correlation Clustering}. In \textsc{Minimum Correlation Clustering},
given a complete graph where each edge is associated with a label in $\{+,-\}$,
the goal is to compute a partition of the vertices of the graph so that
the number of co-clustered vertices joined by $-$ edges and
and the number of vertices joined by $+$ edges and not co-clustered is minimized.
The restriction of \textsc{Minimum Correlation Clustering}
where the output partition has at most $k$ sets, is NP-hard but admits a PTAS~\cite{toc:v002/a013}.
We will extend the analysis of \cite{toc:v002/a013}
by showing that the analogous restriction \textsc{Minimum Consensus Clustering}
admits a PTAS, while being NP-hard.

Notice that  \textsc{Minimum Correlation Clustering} and \textsc{Minimum
  Consensus Clustering} are not comparable, since the input graph in
\textsc{Minimum Correlation Clustering} is unweighted, while the input graph
of \textsc{Minimum Correlation Clustering} is weighted.
On the other hand, it is quite immediate to notice that there are unweighted
graphs that are not an instance of
 \textsc{Minimum Consensus Clustering}.

\section{The problem}
\label{problem}

We will tackle the  \textsc{Consensus Clustering}
problem, in its minimization version. Two elements of the universe set are
\emph{co-clustered} in a partition $\pi$ if they belong to the same set of
$\pi$.

\begin{Def}
\label{symm-dist}
Let $V$ be a universe set and let $\pi_1, \pi_2$ be two partitions of $V$.
Let $d(\pi_1, \pi_2)$ denote the
\emph{symmetric difference distance} defined as 
the number of pairs of elements co-clustered in exactly one
of  $\pi_1$ and $\pi_2$. 
Let $s(\pi_1, \pi_2)$ denote  the \emph{similarity measure}
defined as
the number of pairs of elements co-clustered in both partitions
plus the number of pairs of elements not co-clustered in both
partitions $\pi_1$ and $\pi_2$.
\end{Def}

Given two elements $i, j$ of the universe set $V$ and a set
$\Pi=\{\pi_1,\ldots,\pi_l\}$ of partitions of $V$, we  denote by
$s_{\Pi}(i, j)$ (or simply $s(i, j)$ whenever $\Pi$ is known from
the context)  and the distance $d_{\Pi}(i, j)$  (or simply $d(i, j)$)
respectively,  the number of partitions of $\Pi$ in which $i, j$
are co-clustered and are not co-clustered. Clearly, for each pair
$(i, j)$, $d_{\Pi}(i, j)+s_{\Pi}(i, j)=l$, that is the number of partitions.
When $\Pi$ consists of $2$
partitions $\pi_1$ and $\pi_2$, we denote by $d(\pi_1,\pi_2)$ the quantity
$\sum_{i<j} d_{\{\pi_1,\pi_2\}}(i,j)$.

We are now able to formally introduce the  problem we will study in this paper,
\textsc{Minimum Consensus Clustering} when  the output partition is required to have at most
$k$ sets (denoted by \textsc{k-Min-CC}): we are given a set
$\Pi=\{\pi_1, \pi_2,..., \pi_l\}$ of partitions over universe $V$ and
we want to find  a partition $\pi$ of $V$, such that $\pi$ has at most $k$
sets and $\pi$  minimizes
$d(\pi, \Pi)=\sum_{i=1}^l d(\pi, \pi_i)$, that is the cost of solution $\pi$.
%
%
%
In what follows, we denote by \textsc{k-Min-CC}$(l)$ the restriction of the \textsc{k-Min-CC}
problem where the input consists of exactly $l$ partitions of $V$.

The \textsc{Minimum Consensus Clustering} is closely related to the
\textsc{Minimum Correlation Clustering}~\cite{DBLP:journals/ml/BansalBC04}, where
we are given a labeled complete graph, with each edge
labeled by either $+$ or $-$ and the goal is to compute a partition $C_1, C_2,
\dots ,C_k $ of the vertex set so that
the number of $+$ edges cut by the partition and the number of $-$ edges
inside a same set $C_i$ is minimized.
Several variants of the correlation clustering have been introduced \cite{DBLP:journals/jcss/CharikarGW05,swamy04:_correl_clust,DBLP:journals/jacm/AilonCN08}.

An instance of
\textsc{Minimum Consensus Clustering} can be represented with a labeled
complete graph $G=(V,E)$, where each edge $(v, w) \in E$ is labeled by
$s_{\Pi}(v,w)$. In Section~\ref{PTAS} we assume that the instance of
\textsc{k-Min-CC}$(l)$ is precisely this
graph representation of \textsc{Minimum Consensus Clustering}.

\section{The PTAS}
\label{PTAS}

In this section we will show that the \textsc{k-Min-CC} admits a PTAS, that is
for any $\epsilon>0$ a polynomial time approximation algorithm with a
guaranteed $1+\epsilon$
ratio between the costs of the approximate solution and the optimal solution.
Let $G=(V,E)$ be  the complete graph instance of \textsc{k-Min-CC}.

The MinDisAg algorithm of \cite{toc:v002/a013} for \textsc{Minimum Correlation
Clustering} can be restated to solve \textsc{k-Min-CC} and is reported here as
Alg.~\ref{alg:ptas}.
Let us detail the idea behind MinDisAg  \cite{toc:v002/a013} and how it can be
generalized.
First of all, some ``small'' instances  are solved by a brute force approach,
namely when only one
set must be computed or when the number $n$ of input elements is polynomial
in $k$ (the number of desired output sets).
In fact, there are at most $k^n$ possible partitions of $V$, and $k^n$ is a
constant whenever $n$ is polynomial in $k$.

The algorithm starts by randomly sampling a subset $S$ of $V$. If the sample
is not too large (i.e. $O(\log n)$), then it is
possible to compute all partitions of $S$ in polynomial time.
Since the steps that the algorithm performs for each partition
require polynomial time, the whole algorithm has polynomial time
complexity.

The algorithms extends each partition  of S to
a partition of $V$. Since the number of partitions of $S$ is polynomial, we
can restrict our attention only to the partition  $\tilde{S}$ that fully agrees
on $S$ with the overall optimal solution $\mathcal{D}$.
On that specific partition, extending   $\tilde{S}$ to a partition of $V$
introduces only a few  errors.


More precisely, the algorithm applies a
greedy procedure to  extend $\tilde{S}$:  it  assigns  independently each
element $x$ of $V\setminus S$ to the cluster of $\tilde{S}$ that minimizes the
total cost of all pairs made of $x$ and an element of $S$.

This procedure computes a clustering of $V$ into sets that can be
distinguished into large and small, depending on the fact that a set is
smaller or larger than a certain threshold. The large sets are retained, while
all small sets are merged together obtaining a new universe set which is in
turn recursively fed to the algorithm (only this time requiring  a smaller
error ratio and obtaining a partition with fewer sets.)

We remember that $l$ denotes  the number of input partitions, and
$k$ denotes the number of sets in the output partition. Given a
partition $P$ of $V$, the cost of $P$
is denoted by $cost(P)$.
Let  $\epsilon'$ be equal to $\frac{\epsilon}{128\cdot 20^2 k^4}$ (i.e. $\epsilon'$ is a
constant depending only on $\epsilon$ and the number of sets in the output
partition). We
distinguish two cases: the optimum is at most $\epsilon' n^2$ or
at least $\epsilon' n^2$. In the latter case we exploit the fact
that it is possible to solve the problem in polynomial time and
with a guaranteed \emph{additive} error $\epsilon\epsilon' n^2$,
where $n$ is the number of elements in the universe, for any
constant $\epsilon >0$ (see
\cite{DBLP:journals/jcss/BonizzoniVDJ08} for details). Then the
approximation ratio is at most $\frac{\epsilon\epsilon' n^2 +
  \epsilon' n^2}{\epsilon' n^2} = 1+\epsilon$, that is the algorithm in
\cite{DBLP:journals/jcss/BonizzoniVDJ08} computes the required approximate solution.
Therefore, in the
following we only have to investigate the case when the optimal solution has a cost at most
$\epsilon' n^2$.

We define
$t = \frac{2560000 k^4}{\epsilon^2}\log n$ as the size of the sample set $S$,
$\mathcal{D} = \{\mathcal{D}_1, \ldots, \mathcal{D}_k\}$ as the optimal solution (whose cost is
denoted by $\gamma n^2$).
Let $\tilde{S}$ be the partition $\{X\cap S: X\in \mathcal{D}\}$, that is the restriction of
$\mathcal{D}$ to the set $S$. We recall that we will mainly focus on the iteration of
steps
6--21 where such $\tilde{S}$ is extended to a partition of the universe set $V$.
Let $\mathcal A$ be a partition of a set $A \subseteq V$, and let $x$ be an element of $V$.
Then $N^{\mathcal A}(x)$ is the set of all elements of $A$ different
from $x$ and co-clustered with $x$ in $\mathcal A$.
Given an element $u \in V$, define $val^{\mathcal A}_i(u)$:
\[
val^{\mathcal A}_i(u) = \frac{1}{|A \setminus \{ u \}| } \left(|\{x\in N^{\mathcal A}(u) \wedge s(x,u)=i\}|+|\{x\notin N^{\mathcal A}(u) \wedge d(x,u)=i\}|\right).
\]

Informally $val^{\mathcal A}_i(u)$ is the fraction of pairs consisting of $u$ and an element of $A$ that
may give a contribution $l-i$ to the cost of the  solution. Moreover we define
$val^{\mathcal   A}(u)$ as
\[val^{\mathcal
  A}(u) = \frac{\sum_{x\in N^{\mathcal A}(u)}s(x,u)+\sum_{x
  \notin N^{\mathcal A}(u)}d(x,u)}{l|A  \setminus  \{ u \}| }.\]
Informally  $val^{\mathcal A}(u)$ is
the fraction of input pairs containing $u$ on which  ${\mathcal A}$ agrees.
Notice that $val^{\mathcal A}(u) = \frac{1}{l} \sum_{i=1}^l i\cdot val^{\mathcal A}_i(u)$.
Let $\mathcal A$ be a partition of the set $A$, then
$\mathcal A(u,i)$ is the partition obtained from  $\mathcal A$ moving the element
$u$ to the set $A_i$ (notice that $u$ may not belong to $A$).
Given an integer $j$, with $1\le j\le l$, define $pval^{\mathcal A}(u,i) =
val^{\mathcal A(u,i)}(u)$
and $pval_j^{\mathcal A}(u,i)=val_j^{\mathcal A(u,i)}(u)$.
Finally we introduce the notion of $\beta$-good partition, which  is a good
approximation of the optimal partition. Let $X$ be a subset of $V$, $\mathcal
A$ be a partition of $A$ and
$\beta=\frac{\epsilon}{128\cdot 20^2 k^4}$.
 Then
$\mathcal A$ is
\emph{$\beta$-good} if for each $u\in V$, $0\le j\le l$ and
$1\le i\le k$, then
\[
\left|pval_j^{\mathcal A}(u,i)-pval_j^{\mathcal{D}}(u,i)\right|\le \beta
.\]
%

\begin{algorithm}
\label{alg:ptas}
\KwIn{A set $\Pi$ of partitions of $V$}
\KwOut{A $k$-clustering of the graph, i.e. a partition of $V$ into at most $k$
  sets $V_1,\ldots ,V_k$}
\uIf{$k=1$}{Return the obvious $1$-clustering\;}
\uIf{$n\le 16k^2$}{Return the optimal $k$-clustering, obtained by exhaustive search\;}
ClusMax$\gets$the result of the PTAS for Max Consensus
Clustering~\cite{DBLP:journals/jcss/BonizzoniVDJ08} with accuracy
$\bar{\epsilon}(\epsilon,k)$\;
Pick a sample $S\subseteq V$ by drawing $|S|=\frac{500 \log n}{\beta^2}$
elements uniformly at random  with replacement\;
$m\gets \infty$\;
\ForEach{each partition $\bar{S}$ of $S$,   $\bar{S}=\{S_1, \ldots , S_k\}$}{%
  Initialize the clusters $C_i\gets S_i$ for $1\le i\le k$\;
  \For{each $u\in V\setminus S$}{
    $j_u\gets \text{arg} \min_i\left\{cost\left(\bar{S} \setminus S_i \cup (S_i\cup\{u\})\right)\right\}$\;
    \tcc{$j_u$ maximizes $pval^{\bar{S}}(u,j_u)\}$}
    \tcc{$val^{\bar{S}}(u)\gets pval^{\bar{S}}(u,j_u)$}
    Add $u$ to the set $C_{j_u}$\;
  }
  \tcc{Compute the set of large and small clusters}
  $Large \gets \{j| 1\le j\le k, |C_j|\ge \frac{n}{2k}\}$\;
  $Small \gets \{1,\ldots ,k\}\setminus Large$\;
  $l\gets |Large|$ and $s\gets k-l=|Small|$\;
  $W\gets \bigcup_{j\in Small}C_j$\;
  $\Pi'\gets$ the restriction of the partitions in $\Pi$ to the new universe set $W$\;
  Recursively call MinDisAg on the partitions $\Pi'$ and with arguments
  $(s,\epsilon/3)$. Denote by $W'_1, W'_2, \ldots W'_s$ the result\;
  $\mathcal{C}\gets\{C_1, \ldots, C_l, W'_1, \ldots W'_s\}$\;

  \If{$cost(\mathcal{C})< m$}{
    $m\gets cost(\mathcal{C})$\;
    ClusMin$\gets\mathcal{C}$\;
  }
}
Return the better of the two clusterings ClusMax and ClusMin\;
\caption{MinDisAg$(k,\epsilon)$}
\end{algorithm}

\subsection{Analysis of the Algorithm}


Notice that the main contribution of this section lies in Lemma
\ref{lemma:1} which is a stronger version of a result in
\cite{toc:v002/a013}; in that paper the notion of $pval^A(u,i)$ is
sufficient because the problem studied is
unweighted. In our paper we study a problem where each pair of elements can
have a cost that is an integer between $0$ and $l$, therefore we
need a definition of $pval_j^A(u,i)$, with a new parameter
$j$ expressing the number of input partitions where two elements
are either co-clustered or not co-clustered. Indeed our definition of $\beta$-goodness
requires that a certain inequality holds for values of $j$ that
are integers between $0$ and $l$, while  in
\cite{toc:v002/a013}   $j$ can -- implicitly -- only
take $0$ or $1$ as value.

Recall that we denote by $\tilde{S}$ the restriction of $\mathcal{D}$ to the sample set $S$.
The following lemma proves that $\tilde{S}$ is, with high probability, a good sample of the
optimal solution.

\begin{Lem}
\label{lemma:1}
The partition  $\tilde{S}$ is $\beta$-good with probability at least $1-O(\frac{1}{\sqrt{n}})$.
\end{Lem}

\begin{proof}
Let $v$ be an element of $S$ and let $u$ be an element of $V$. Let $p(v,i,j)$
be a variable equal to $1$ if and
only if $v\in N_{\mathcal{D}(u,i)}(u)$ and $s(v,u)=j$ or $v\notin N_{\mathcal{D}(u,i)}(u)$
and $d(v,u)=j$. Pose $p(v,i,j)=0$  otherwise.

By construction of $p(v,i,j)$ and $pval_j^{\mathcal{D}}(u,i)$, the probability
$Pr[p(v,i,j)=1] = pval_j^{\mathcal{D}}(v,i)$, as the set $S$ is sampled
randomly from $V$. Also notice that $pval_j^{\tilde{S}}(v,i)=val_j^{\tilde{S}(v,i)}(v)=
\frac{1}{|S\setminus \{u\}|} \left(|\{x\in N_{\tilde{S}(v,i)}(v) \wedge s(x,v)= j\}|+|\{x\notin N_{\tilde{S}(v,i)}(v)
\wedge d(x,v)=j\}|\right)=\frac{1}{|S\setminus \{u\}|} \sum_{v\in S \setminus \{u\}}
p(v,i,j)$, as the latter equality is an immediate consequence of the
definition of $p(v,i,j)$.

The Hoeffding bound states that, given some causal variables  $X_i$ such
that $Pr[X_i=1]=p$ (and $X_i=0$ otherwise), then
$Pr[| X_a - \frac{1}{m}\sum_{a=1}^m X_a|>\beta]\le 2e^{-2m\beta^2}$.
In our case the causal variable $X_a$ are $p(v,i,j)$, and the sum is over all
elements $v\in S\setminus \{u\}$, therefore the inequality becomes
$Pr[| p(v,i,j) - \frac{1}{|S\setminus \{u\}|}\sum_{v\in S\setminus \{u\}} p(v,i,j)|>\beta]\le
2e^{-2(|S\setminus \{u\}|)\beta^2} \le 2e^{-2t\beta^2}$.
By the previous arguments, the inequality can be rewritten as
$Pr[| pval_j^{\mathcal{D}}(u,i) - pval_j^{\tilde{S}}(u,i) |>\beta]\le
2e^{-2t\beta^2}$, which gives an upper bound on the probability that any
element $u\in V$ does not satisfy the requirements of an $\beta$-good set.

Applying a union bound we obtain that the probability of having at least one
of the $t$ elements not satisfying the requirements is at most
$2t e^{-2t\beta^2}$. Since $|S|=\frac{500\log n}{\beta^2}$, the partition
$\tilde{S}$ is $\beta$-good with probability at least $1-2 \frac{500\log
  n}{\beta^2} e^{-1000\log n}=
1-2 \frac{500\cdot 160^2\cdot 20^2 k^4 l^2 \log
  n}{\beta\epsilon^2}\frac{1}{n^{1000}}$, which is trivially larger than
$1-\frac{c}{\sqrt{n}}$
for some constant $c$.
\end{proof}

We will now provide some simple generalizations of the Lemmas in
\cite{toc:v002/a013}, omitting the proofs as they are straightforward
extensions of those in \cite{toc:v002/a013}. Just as in  \cite{toc:v002/a013},
we will assume that the sample $S$ is $\beta$-good, for some constant $\beta$, and we will focus on the
iteration of the algorithm for the partition $\bar{S}$ of $S$ that agrees with
the optimal partition ${\mathcal \mathcal{D}}$. We will denote by $C_1, \ldots , C_k$ the
sets in ClusMin at the end of such iteration.


\begin{Lem}[Lemma 4.3 in \cite{toc:v002/a013}]
\label{lemma:sample-set-is-ok}
Let $u\in V \setminus S$ with $u\in \mathcal{D}_s$ (that is the $s$-th set of the optimal
solution), and $u\in C_r$ for $r\neq s$ (that is  $u$ is misplaced by the
algorithm).
Then $pval_j^{\mathcal{D}}(u,r)\ge pval_j^{\mathcal{D}}(u,s)-2\beta=val_j^{\mathcal{D}}(u)-2\beta$
for each  $0\le j\le k$.
\end{Lem}


Recall that $l$ is the number of input partitions, define
$T_{low}$ as the
set $\{ u\in V : val^{\mathcal{D}}(u)\le 1-\frac{1}{20k^2}\}$, and let us
call \emph{bad} all elements in $T_{low}$ and  \emph{good}
all elements that are not in $T_{low}$.
As  each element $u$ in $T_{low}$ contributes to the cost of a solution of \textsc{k-Min-CC}$(l)$
for at least $\frac{1}{2}l(n-1)(1- val^{\mathcal{D}}(u))\le \frac{1}{40k^2}l(n-1)$,
a simple counting argument allows us to prove
that there are at most $\frac{80\gamma n k^2}{l}$ bad elements.

For clarity's sake, we split Lemma~4.4 in \cite{toc:v002/a013} into two
separate statements, where the first statement
(Lemma~\ref{lemma:different-large-and-good-are-small}) is actually proved in
the first part of the proof of Lemma~4.4 in \cite{toc:v002/a013}, while the
second statement corresponds to Lemma~4.4  in \cite{toc:v002/a013}.
Those  technical results show that (i) our algorithm  clusters almost
optimally all good elements and (ii) all good elements in $Large$ are
optimally clustered, pending a  condition on various
parameters that  will be  proved at the end of the section
(for the definition of $Large$ and $Small$  see Algorithm \ref{alg:ptas}).
More precisely,
Lemma~\ref{lemma:different-large-and-good-are-small} states that  misplaced
good elements must
belong to some small sets (which in turn implies that the majority of good
elements must be optimally clustered).

\begin{Lem}
\label{lemma:different-large-and-good-are-small}
Let $u$ be an element in $C_i \setminus T_{low}$ but not in $\mathcal{D}_i\setminus T_{low}$. Then $u\in
\mathcal{D}_j$, for some $j \neq i$, and $|\mathcal{D}_i|\le 2(\frac{1}{20k^2}+\beta)n+1$.
\end{Lem}

\begin{proof}
The proof is the same as in \cite{toc:v002/a013}, except for the observation
that, by our definition of $pval$ and since each pair
of elements involving $u$ is correctly co-clustered when $u$ is in either $\mathcal{D}_i$ or $\mathcal{D}_j$,
$pval^{\mathcal{D}}(u,j) + pval^{\mathcal{D}}(u,i)\le 2 - \frac{l(|\mathcal{D}_i|+|\mathcal{D}_j|-1)}{l(n-1)}$.
\end{proof}

\begin{Lem}
\label{lemma:good-are-optimal}
Let $i$ be an element in $Large$. If $\frac{n}{2k} - \gamma
n^2\frac{40k^2}{l(n-1)}  > 2(k+1)\left((\frac{1}{20k^2}+\beta)n+1\right)$ and
$2(\frac{1}{20k^2}+\beta)n+k < \frac{n}{2k} - \frac{80\gamma n k^2}{l}$
then $C_i\setminus T_{low}=\mathcal{D}_i\setminus T_{low}$.
\end{Lem}

\begin{proof}
  Let  $x\in V\setminus T_{low}$. W.l.o.g. we can assume that $x\in C_1\setminus T_{low}$ and $x\in
\mathcal{D}_1 \cup T_{low}$. First we will prove that $C_1\setminus T_{low} \subseteq \mathcal{D}_1\setminus T_{low}$.
Assume to the contrary that there exists a $y \in C_1$, $y\notin \mathcal{D}_1, T_{low}$,
therefore (w.l.o.g.) $y\in \mathcal{D}_2$.
By Lemma~\ref{lemma:different-large-and-good-are-small}, and since
there are at most $k$ sets in $\mathcal{D}$, then $|C_1 \setminus (\mathcal{D}_1\cup T_{low})|\le
2(\frac{1}{20k}+\beta k)n+k$.

Since $C_1 \setminus (\mathcal{D}_1\cup T_{low}) = (C_1  \setminus
\mathcal{D}_1) \setminus  T_{low} = (C_1 \setminus  T_{low})  \setminus
\mathcal{D}_1$
then $|\mathcal{D}_1|\ge |C_1\setminus T_{low}| \setminus |C_1 \setminus (\mathcal{D}_1\cup T_{low})|
\ge |C_1| \setminus |T_{low}| \setminus |C_1 \setminus (\mathcal{D}_1\cup T_{low})| \ge
\frac{n}{2k} - \gamma n^2\frac{40k^2}{l(n-1)}  - 2(\frac{1}{20k}+\beta k)n
+k$. But
$\frac{n}{2k} - \gamma n^2\frac{40k^2}{l(n-1)}  - 2(\frac{1}{20k}+\beta
k)n+k> 2(\frac{1}{20k^2}+\beta)n+1$,
 which
contradicts  $|\mathcal{D}_1| \le 2(\frac{1}{20k^2}+\beta)n+1$.
In fact
$\frac{n}{2k} - \gamma n^2\frac{40k^2}{l(n-1)}  - 2(\frac{1}{20k}+\beta
k)n+k> 2(\frac{1}{20k^2}+\beta)n+1$
can be rewritten as
$\frac{n}{2k} - \gamma n^2\frac{40k^2}{l(n-1)}  > 2(k+1)\left((\frac{1}{20k^2}+\beta)n+1\right)$.

Now we know that $C_1\setminus T_{low} \subseteq \mathcal{D}_1\setminus T_{low}$ and we would like to prove that
$C_1\setminus T_{low} \supseteq \mathcal{D}_1\setminus T_{low}$, along the same lines as for the first part.
Assume to the contrary that there exists a $y \in \mathcal{D}_1$, $y\notin C_1, T_{low}$,
therefore (w.l.o.g.) $y\in C_2$.
Again by Lemma~\ref{lemma:different-large-and-good-are-small}, both $\mathcal{D}_1$
and $\mathcal{D}_2$ have at most $2(\frac{1}{20k^2}+\beta)n+1$ elements.
Notice that $C_1 \setminus T_{low} \subseteq  \mathcal{D}_1$, since  $C_1\setminus T_{low} \subseteq \mathcal{D}_1\setminus T_{low}$,
moreover $C_1$ is large, therefore $|C_1|\ge \frac{n}{2k}$. By the value
of $|T_{low}|$, $2(\frac{1}{20k^2}+\beta)n+k \ge \frac{n}{2k} - \frac{80\gamma
  n k^2}{l}$ which does not hold by hypothesis.
\end{proof}

Now we are able to show that there is a solution where some sets are exactly
the large sets in ClusMin and whose cost is not much larger than the optimum.
This fact justifies the recursive step of the algorithm. The condition under
which the lemma holds will be proved at the end of the section.

\begin{Lem}
\label{lem:large-sets-can-be-extended}
If
$l(n-1)|T_{low}|\left( 2\beta  +\frac{|T_{low}|}{l(n-1)}\right)\le \frac{\epsilon}{2}\gamma n^2$, then
there exists a
solution $F=\{F_1, \ldots , F_k\}$ such that the cost of $F$ is at most
$\gamma n^2(1+\epsilon/2)$ and $F_i=C_i$ for each $i$ in Large.
\end{Lem}

\begin{proof}
Let $F$ be the solution consisting of all large sets in ClusMin and where all remaining
elements are partitioned as in $\mathcal{D}$. Clearly the only pairs of elements that
might not be  partitioned in $F$ as in ClusMin are the ones containing at least
one element of $T_{low}$, by Lemma~\ref{lemma:good-are-optimal}. By the
definition of $val$, $cost(F)-cost(\mathcal{D})\le l(n-1)\sum_{x\in
  T_{low}}\left(val^{\mathcal{D}}(x) - val^{F}(x)\right)$.

We have to consider two different cases, depending on the fact that $x\in T_{low}$
belongs to sets $C_i$, $\mathcal{D}_i$ for a certain $i$, or not. In the first
case w.l.o.g. $x$ is in both $C_1$ and $\mathcal{D}_1$
the set of pairs that are different in ClusMin and in $\mathcal{D}$, are only pairs of
the form $(x,y)$ where $y\in T_{low}$, which in turn implies that $val^{F}(x) \ge
val^{\mathcal{D}}(x) - \frac{|T_{low}|}{l(n-1)}$. In the second case we can assume
w.l.o.g. that $x\in C_1$ and $x\in \mathcal{D}_2$. Applying
Lemma~\ref{lemma:sample-set-is-ok} we know that $pval^{\mathcal{D}}(x,y)\ge
val^{\mathcal{D}}(x)- 2\beta$. Also notice that in
$\mathcal{D}(x,1)$ and $F$, the element $x$
belong to the same set therefore, just as for the first case, $val^{F}(x) \ge
val^{\mathcal{D}(x,2)}(x) - \frac{|T_{low}|}{l(n-1)}$, but $val^{\mathcal{D}(x,2)}(x) =
pval^{\mathcal{D}}(x,2)$. Combining all inequalities we obtain $val^{F}(x) \ge
pval^{\mathcal{D}}(x,2) - \frac{|T_{low}|}{l(n-1)} \ge val^{\mathcal{D}}(x) - 2\beta
-\frac{|T_{low}|}{l(n-1)}$, where the last inequality comes from
Lemma~\ref{lemma:sample-set-is-ok}. In both cases we can say that $val^{F}(x) \ge
pval^{\mathcal{D}}(x,2) - \frac{|T_{low}|}{l(n-1)} \ge val^{\mathcal{D}}(x) - 2\beta
-\frac{|T_{low}|}{l(n-1)}$.
An immediate consequence is that  $cost(F)-cost(\mathcal{D})$ is at most
$l(n-1)\sum_{x\in T_{low}}\left(val^{F}(x) - val^{\mathcal{D}}(x)\right) \le l(n-1)
|T_{low}|\left( 2\beta  +\frac{|T_{low}|}{l(n-1)}\right)$. The claim follows since
$l(n-1)
|T_{low}|\left( 2\beta  +\frac{|T_{low}|}{l(n-1)}\right)\le \gamma n^2\epsilon/2$.
\end{proof}

Since the partitions $F$ and ClusMin are the same for all pairs where at least
one element is in a large set of ClusMin, an immediate consequence is that the
solution returned by the algorithm has cost at most
$\gamma n^2(1+\epsilon/3)(1+\epsilon/2)$ which is at most equal to $\gamma n^2(1+\epsilon)$
for any sufficiently small $\epsilon$.
%
%
%
The following technical result completes our proof by showing that
Lemma~\ref{lem:large-sets-can-be-extended} holds. The proof is a  mechanical
consequences of the values of $\beta$, $|T_{low}|$ and $\epsilon'$.

\begin{Lem}
\label{lemma:2}
$l(n-1)
|T_{low}|\left( 2\beta  +\frac{|T_{low}|}{l(n-1)}\right)\le \frac{\epsilon}{2} \gamma n^2$.
\end{Lem}

\begin{proof}
Since $|T_{low}|\le \frac{40\gamma nk^2}{l}$ and
$\beta=\frac{\epsilon}{20\cdot 160 k^2 l}$, it suffices to prove that
$l(n-1)\frac{40\gamma nk^2}{l}\left( \frac{\epsilon}{20\cdot 80k^2l}
  + \frac{40\gamma nk^2}{l^2(n-1)}\right)\le \frac{\epsilon}{2}
\gamma n^2$ that is equivalent to
$\frac{80 (n-1)k^2}{l}\left( \frac{\epsilon}{20\cdot 80k^2} +
  +\frac{40k^2 \gamma n}{l(n-1)}\right)\le \epsilon n$.
Since we are only interested in instances where the algorithm of
\cite{DBLP:journals/jcss/BonizzoniVDJ08} fails to provide a $(1+\epsilon)$
approximation ratio, we can assume that $\gamma <
\epsilon'=\frac{\epsilon}{128\cdot 20^2 k^4}$, consequently it suffices to
prove that
$\frac{80(n-1)k^2}{l}\left( \frac{\epsilon}{20\cdot 80 k^2 } + \frac{2 \cdot 20 k^2
    \epsilon n }{128\cdot 20^2 l (n-1) k^4}\right)\le \epsilon n$
that is equivalent to
$\frac{4(n-1)}{l}\left( \frac{1}{80} + \frac{n}{64  l (n-1) }\right)\le n$
which in turn is equivalent to $9n \le 80ln+4n$ which is trivially true.
\end{proof}

To complete the section and the analysis of the algorithm, we need to prove
that the assumptions that we have made in some of the previous lemmas actually
hold. The proofs are mechanical and quite tedious consequences of the values
of $\beta$, $\gamma$ and $\epsilon'$.

\begin{Lem}
\label{lemma:good-are-optimal-assumption}
If $n\ge 16k^2$ then
$\frac{n}{2k} - \gamma n^2\frac{40k^2}{l(n-1)}  > 2(k+1)\left((\frac{1}{20k^2}+\beta)n+1\right)$.
\end{Lem}

\begin{proof}
By the values of $\gamma$ and $\beta$, and since we can assume that $\gamma <
\epsilon'=\frac{\epsilon}{128\cdot 20^2 k^4}$, the inequality can be rewritten as
$\frac{n}{2k} - \frac{\epsilon}{128\cdot 20^2 k^4} n^2\frac{40k^2}{l(n-1)}  >
2(k+1)\left((\frac{1}{20k^2}+\frac{\epsilon}{128\cdot 20^2 k^4})n+1\right)$
which can be simplified as
$\frac{n}{2k} - \frac{\epsilon}{64\cdot 20 k^2} n^2\frac{1}{l(n-1)}  >
2(k+1)\left((\frac{1}{20k^2}+\frac{\epsilon}{128\cdot 20^2 k^4})n+1\right)$.
Since $\frac{n}{n-1}\le 2$, it suffices to prove that
$\frac{n}{2k}\left(1 - \frac{\epsilon }{16\cdot 20 k l} \right) >
2(k+1)\left((\frac{1}{20k^2}+\frac{\epsilon}{128\cdot 20^2 k^4})n+1\right)$.
As $k,l\ge 2$ and $\epsilon$ is tiny, $\frac{\epsilon }{16\cdot 20 k l}<
\frac{1}{1000}$, therefore we are only interested in proving that
$\frac{999}{1000}\cdot \frac{n}{2k}>
2(k+1)\left((\frac{1}{20k^2}+\frac{\epsilon}{128\cdot 20^2 k^4})n+1\right)$
which is equivalent to
$\frac{999}{1000}\cdot \frac{n}{2k}>
(k+1)(\frac{1}{10k^2}+\frac{\epsilon}{64\cdot 20^2 k^4})n+2(k+1)$. Again,
$\frac{k+1}{k}\le 2$, therefore it is sufficient to prove that
$\frac{999}{1000}\cdot \frac{n}{2k}>
(\frac{1}{5k}+\frac{\epsilon}{32\cdot 20^2 k^3})n+2(k+1)$ which is equivalent
to
$\frac{599}{1000}\cdot \frac{n}{2k}>
\frac{\epsilon}{32\cdot 20^2 k^3}n+2(k+1)$. Since $k\ge 2$,
$\frac{\epsilon}{32\cdot 20^2 k^3}<\frac{1}{1000}$, hence  it suffices to
prove that $\frac{n}{4k}>2(k+1)$ which is an immediate consequence of the
assumption $n\ge 16k^2$.
\end{proof}

\begin{Lem}
\label{lemma:good-are-optimal-assumption2}
If $n\ge 16k^2$ then
$2(\frac{1}{20k^2}+\beta)n+k < \frac{n}{2k} - \frac{80\gamma n k^2}{l}$.
\end{Lem}

\begin{proof}
By the values of $\gamma$ and $\beta$ the inequality can be rewritten as
$2(\frac{1}{20k^2}+\frac{\epsilon}{128\cdot 20^2 k^4})n+k < \frac{n}{2k} -
\frac{80 n k^2}{l}\frac{\epsilon}{128\cdot 20^2 k^4}$ which can be simplified
as
$\frac{n}{k}\left( \frac{1}{5k}+\frac{\epsilon }{32\cdot 20^2 k^3}+
  \frac{\epsilon }{16\cdot 20 lk}\right)+2k < \frac{n}{k}$. As $k,l\ge 2$, it
is immediate to notice that $\frac{1}{5k}+\frac{\epsilon }{32\cdot 20^2 k^3}+
  \frac{\epsilon }{16\cdot 20 lk}\le \frac{1}{4}$, therefore it suffices to
  prove that $2k<\frac{3n}{4k}$, which is an immediate consequence of the
assumption $n\ge 16k^2$.
\end{proof}

\section{NP-hardness}
\label{NP-hard-proof}

In this section we prove that \textsc{2-Min-CC}($3$) 
is NP-hard. From the NP-hardness of \textsc{2-MIN-CC},
it is easy to  show that also \textsc{k-MIN-CC}($3$) is NP-hard for any fixed $k$.
Our proof consists of a reduction from
the NP-hard Min Bisection Problem (MIN-BIS) to \textsc{2-MIN-CC}($3$).
The MIN-BIS problem, given a graph $G=(V,E)$, asks for a partitioning of $V$
in two equal-sized sets
so that the number of edges connecting vertices in different sets is minimized.

For our purposes, in this section we give a different, but equivalent,
definition of  cost  of a solution $\pi$ of
\textsc{Minimum Consensus Clustering} over instance $\Pi$ can be
alternatively defined as:
\begin{equation}\label{rel}
\sum_{\forall (i<j)} (r_{\pi}(i,j) d_{\Pi}(i, j) +
(1-r_{\pi}(i,j)) s_{\Pi}(i, j)),
\end{equation}
where $ r_{\pi}(i,j) = 1$ iff  $(i,j)$ are co-clustered in $\pi$,
otherwise $ r_{\pi}(i,j) = 0$. The above formula will be used in
the paper (see Section 4) to define the cost of a set $P$ of pairs
in a solution $\pi$  as $\sum_{\forall (i,j)\in P} (r_{\pi}(i,j)
d_{\Pi}(i, j) + (1-r_{\pi}(i,j)) s_{\Pi}(i, j))$.

Given an instance $G=(V,E)$  of MIN-BIS, where $|V|=n$ and $|E|=m$,
we build an instance of  \textsc{2-MIN-CC}($3$) as follows.

First we define the universe set $V$.
For each $v_i \in V$, we define a set of $n^4$ elements $X_i=\{ x_{i,1},\dots x_{i,n^4} \}$,
and a set of $n$ elements $Y_{i}=  \{ y_{i,1}, \dots, y_{i,n} \}$.
The universe set is $V=(\cup_i X_i \cup Y_i)$.
Next we define the three input partitions of  \textsc{2-MIN-CC}($3$),
$\Pi=\{ \pi_1, \pi_2 ,\pi_3 \}$.
Partitions $\pi_1$ and $\pi_2$ are identical and consist of $n$ disjoint sets
$ X_i \cup Y_i $, with $i=1,\ldots, n$.
The partition $\pi_3$ contains the sets $X_i$, moreover
for each edge $(v_i,v_j)\in E$, in $\pi_3$ we have the set $\{y_{i,h},
y_{j,l}\}$ consisting of two elements taken respectively from $Y_i$ and $Y_j$
(the actual elements taken are not important, provided that $\pi_3$ is a
partition of the universe set -- which is trivial to obtain).
Finally,  in $\pi_3$ we have a singleton  for each element of $\cup Y_i$ that
are not in a two-element set according to the previous rule.

\begin{Obs}
\label{obs:3-sets}
Since all the elements in $X_i$ are co-clustered in all input partitions,
each $X_i$ is contained in a set of the optimal solution.
\end{Obs}

The previous observation allows ourselves to restrict our attention to
solutions where all elements of $X_i$ are co-clustered.
Consider a solution $\pi=(S_1, S_2)$. The cost of
$\pi$ can be expressed as the cost of all pairs of elements in $\pi$. We can
split the cost of $\pi$ into four parts:
\begin{enumerate}
\item the cost of pairs of elements both belonging to $\cup X_i$,
\item the cost of pairs of elements with exactly one element belonging to $\cup X_i$,
\item the cost of pairs of elements in $Y_i\times Y_j$ with $i\neq j$,
\item the cost of pairs of elements both belonging to a set $Y_i$.
\end{enumerate}

We will call
\emph{balanced} a solution $(S_1, S_2)$ where both $S_1$ and $S_2$ contain
exactly $\frac{n}{2}$ sets $X_i$. The following lemma states that optimal
solutions must be balanced.

\begin{Lem}
\label{lem-NP-eq-size}
Let  $\pi=(S_1, S_2)$ be a  solution  of \textsc{2-MIN-CC}($3$), then the cost of
$\pi$ is at most  $\frac{3}{4}n^{10} - \frac{3}{2} n^9 + 3n^7 + \frac{3}{2}
n^4- \frac{3}{2} n^3$ if and only if $\pi$ is a balanced solution.
\end{Lem}
\begin{proof}
Notice that the total cost of case 2) is at most $3n^2 \cdot n^5=3n^7$ as
$|\cup Y_i|=n^2$ and  $|\cup X_i|=n^5$, while the sum of total costs of
cases 3) and 4) is at most $3{{n^2} \choose 2} = \frac{3}{2} n^4- \frac{3}{2} n^3$.

Let $z$ be the number of sets $X_i$ included in $S_1$.
The cost of the pairs of elements both belonging to $\cup X_i$ is
$C(z)=3\left ( {z \choose 2} + {{n-z}\choose 2}\right )n^8$.
Indeed,  only
the pairs of elements in  distinct sets $X_i$ that are
co-clustered in $S_1$ and $S_2$ contribute to the cost, as no pair
of elements belonging to two distinct sets $X_i$ is co-clustered
in an input partition.
The minimum of $C(z)$ is attained for $z=\frac{n}{2}$. For any other $z$, the value of $C(z)$
is at least
equal $C(\frac{n}{2}-1)$.

Since $C(\frac{n}{2})=\frac{3}{4}n^{10} - \frac{3}{2} n^9$, the maximum total cost for a
balanced solution is $\frac{3}{4} n^{10} - \frac{3}{2} n^9 + 3n^7 +
\frac{3}{2} n^4- \frac{3}{2} n^3$, while the maximum total
cost for an unbalanced solution is  at least
$C(\frac{n}{2}-1)=\frac{3}{4} n^{10} -\frac{3}{2} n^9 +3n^86>
\frac{3}{4}n^{10} - \frac{3}{2} n^9 + 3n^7 + \frac{3}{2} n^4- \frac{3}{2} n^3$.
\end{proof}

From Lemma~\ref{lem-NP-eq-size} we can consider only balanced solutions. A balanced
solution $\pi$ is called \emph{standard} if, for each $i$, $X_i$ and $Y_i$ are
contained in the same set of $\pi$. The following lemma shows that we can
consider only standard solutions

\begin{Lem}
\label{lem-NP-el_co_cluster}
Let  $\pi=(S_1, S_2)$ be a balanced solution  of \textsc{2-MIN-CC}($3$), then the cost of
$\pi$ is at most  $\frac{3}{4}n^{10} -\frac{3}{2} n^9+ \frac{3}{4}n^7 - \frac{1}{2}n^6 +\frac{1}{4} n^4 + \frac{1}{2} n^3 - \frac{1}{2}n^2$ iff  $\pi$ is a
standard solution.
\end{Lem}

\begin{proof}
Let $\pi=(S_1, S_2)$ be a balanced solution, then the total cost  of pairs of
elements with exactly one element belonging to $\cup X_i$ is at most
$\frac{3}{4}n^7 - \frac{1}{2}n^6$ as all pairs in $X_i\times Y_j$, with $i\neq j$, contribute with a
cost $3$ if and only if $X_i\cup Y_j$ is contained in a set of $\pi$, and have no
cost otherwise. At the same time all pairs in $X_i\times Y_i$ have cost $1$ in
any standard solution, as $X_i\cup Y_i$
are a set of two input partitions, while in the third input partition,
$\pi_3$, no pairs in  $X_i\times Y_i$ are co-clustered.
If $\pi$ is a standard solution, then the total cost of pairs of elements
in $Y_i\times Y_j$ with $i\neq j$ is $\frac{1}{4} n^4$ as only half of such pairs are
co-clustered in a standard solution.
Following the reasoning of the  proof of Lemma~\ref{lem-NP-eq-size}, with our
new estimates of cases 2) and 3), it is immediate to notice
that, if $\pi$ is a standard solution, then its cost is at most
$\frac{3}{4} n^{10} - \frac{3}{2} n^9+ \frac{3}{4}n^7 - \frac{1}{2}n^6 +\frac{1}{4} n^4 + n {n\choose 2}=
\frac{3}{4} n^{10} - \frac{3}{2} n^9+ \frac{3}{4}n^7 - \frac{1}{2}n^6 +\frac{1}{4} n^4 + \frac{1}{2} n^3 - \frac{1}{2}n^2$.

Now assume that $\pi$ is not a standard solution, that is there exists an
element $y\in Y_i$ that is not clustered together with all elements of $X_i$.
Again, following the same lines of the proof of Lemma~\ref{lem-NP-eq-size},
the cost of $\pi$ is at least $\frac{3}{4}n^{10} - \frac{3}{2} n^9+
\frac{3}{4}n^7 - \frac{1}{2}n^6 + \frac{1}{4} n^4 + n^4$, as all pairs in
$\{y\}\times X_i$ have a cost $2$, instead of $1$ as in a standard partition.
Since $\frac{3}{4}n^{10} - \frac{3}{2} n^9+ \frac{3}{4}n^7 - \frac{1}{2}n^6 +
\frac{1}{4} n^4 + n^4 > \frac{3}{4}n^{10} - \frac{3}{2} n^9+ \frac{3}{4}n^7 -
\frac{1}{2}n^6 +\frac{1}{4} n^4 + \frac{1}{2} n^3 - \frac{1}{2}n^2$,  the
lemma follows.
\end{proof}

Given a standard solution $\pi$, by construction of the reduction, with each edge $(v_i,v_j)\in E$, we associate
a pair $\{y_{i,h}, y_{j,l}\}$. Let us denote by $F$ the set of such pairs, and
by $F_c$ the subset of all pairs in $F$ that are co-clustered in $\pi$.
We conclude the proof with the following theorem.

\begin{Thm}
\label{Thm-NP}
Let $G=(V,E)$ be an instance of MIN-BIS, and let $(\pi_1, \pi_2, \pi_3)$ be its
associated instance of \textsc{2-MIN-CC}($3$). Then  $(\pi_1, \pi_2, \pi_3)$
has a solution of cost
$\frac{3}{4}n^{10} - \frac{3}{2} n^9+ 3/4n^7 - \frac{1}{2}n^6 +\frac{1}{4} n^4
+ \frac{3}{2} n^4 - \frac{1}{2} n^3 - \frac{1}{2}n^2 + (|F|-k) -k$ if and only
if $G$ has a bisection of cost $k$.
\end{Thm}
\begin{proof}
Let $(V_1, V_2)$ be a bisection with cost $k$. Then let $S_1$ be the set
$\cup_{i\in V_1} (X_i\cup Y_i)$, and let $S_2=\cup_{i\in V_2} (X_i\cup Y_i)$.
By construction $(S_1, S_2)$ has cost $\frac{3}{4}n^{10} - \frac{3}{2} n^9+ \frac{3}{4}n^7 - \frac{1}{2}n^6
+\frac{1}{4} n^4 + \frac{3}{2} n^4 - \frac{1}{2} n^3 - \frac{1}{2}n^2 + (|F|-k) -k$.

Now let $(S_1, S_2)$ be a solution of  \textsc{2-MIN-CC}($3$) with cost
$\frac{3}{4} n^{10}/4 - \frac{3}{2} n^9+ \frac{3}{4}n^7 - \frac{1}{2}n^6
+\frac{1}{4} n^4 + \frac{3}{2} n^4 - \frac{1}{2} n^3 - \frac{1}{2}n^2 + (|F|-k) -k$. By
Lemmas~\ref{lem-NP-eq-size},~\ref{lem-NP-el_co_cluster} $(S_1, S_2)$ must be
a standard solution.

Recall that the cost of a solution $\pi=(S_1, S_2)$ can be
expressed as the cost of all pairs of elements in $\pi$, such a
cost can be split into parts 1), 2), 3) and 4).
Moreover, following the proof of
  Lemmas~~\ref{lem-NP-eq-size},~\ref{lem-NP-el_co_cluster}, we know that the
total cost of case 1) is $\frac{3}{4}n^{10} - \frac{3}{2} n^9$, the total cost of case 2)
is $\frac{3}{4}n^7 - \frac{1}{2}n^6 $. By direct inspection the total cost of case 4) is
$\frac{1}{2}n^3 + \frac{1}{2}n^2$.

We still have to consider case 3), that is the cost of  pairs $(y_{i,q},
y_{j,t})$, with $j \neq i$.
We have to distinguish three cases, according to the fact that
$(y_{i,q},y_{j,t})\in F-F_c$ (in this case the cost is $1$),
$(y_{i,q},y_{j,t})\in F_c$  (in this case the cost is $2$),
$(y_{i,q},y_{j,t})\notin F$ (in this case the cost is $3$ if $y_{i,q}$ and
$y_{j,t}$ are co-clustered, and $0$ otherwise.
Therefore the total cost of case 3) can be written as $n^2 {n \choose
  2}+|F-F_c|-|F_c|$.

Summing up the costs of the four cases we obtain a total cost $\frac{3}{4}n^{10} - \frac{3}{2}
n^9+ \frac{3}{4}n^7 - \frac{1}{2}n^6 +\frac{1}{4}n^4 + \frac{3}{2} n^4 -
\frac{1}{2} n^3 - \frac{1}{2}n^2 + |F| -2|F_c|$.
Consequently, taking into account the initial hypothesis,  $|F_c|=k$.
Let $(V_1, V_2)$ be the solution of $G$ where $V_1=\{ v_i | X_i\subseteq
S_1\}$ and $V_2=V-V_1$. By construction the number of edges of $E$ crossing
the bipartition  $(V_1, V_2)$ is equal to $|F_c|$ which, in turn, is equal to
$k$ completing the proof.
\end{proof}

\section{Conclusions}

In this paper we have studied the \textsc{Minimum Consensus Clustering}
problem when
the output partition contains at most a constant number of sets. We have shown
that the MinDisAg algorithm~\cite{toc:v002/a013} can be applied also for our
problem, hence showing that its applicability is not restricted to unweighted
problems. Moreover we have proved that the same problem is NP-hard even on instances of three
input partitions, thereby justifying our reliance on polynomial time
approximation algorithms.

In our opinion the main idea behind MinDisAg algorithm could be applied to
some more general versions of both \textsc{Minimum Consensus Clustering} and
\textsc{Minimum Correlation Clustering}  than the ones studied here and
in~\cite{toc:v002/a013}.

\section*{Acknowledgments}

PB, and GDV  have been partially supported by FAR 2008 grant
  ``Computational models for phylogenetic analysis of gene variations''.
PB has been partially supported by
the MIUR PRIN 2007 Project  ``Mathematical aspects and emerging
applications of automata and formal languages''.


\end{document}